\documentclass[letterpaper]{article}
\usepackage[margin=1.00in]{geometry}
\usepackage{mathtools}
\usepackage{amsmath,amsfonts,amssymb,amsthm,url,verbatim,cite}
\usepackage{graphicx,subcaption}
\usepackage{enumitem}
\usepackage{authblk}
\usepackage[ pdftex, plainpages = false, pdfpagelabels,
                 pdfpagelayout = useoutlines,
                 bookmarks,
                 bookmarksopen = true,
                 bookmarksnumbered = true,
                 breaklinks = true,
                 linktocpage=all,
                 pagebackref=false,
                 colorlinks = true,
                 linkcolor = BrickRed,
                 urlcolor  = blue,
                 citecolor = BrickRed,
                 anchorcolor = green,
                 hyperindex = true,
                 hyperfigures
                 ]{hyperref}
\usepackage[usenames, dvipsnames]{xcolor}
\usepackage{xifthen} 

\DeclarePairedDelimiter\ket{\lvert}{\rangle}

\DeclarePairedDelimiterX\braket[2]{\langle}{\rangle}{#1 \delimsize\vert #2}
\DeclarePairedDelimiterX\bbraket[2]{\langle\!\langle}{\rangle\!\rangle}{#1 \delimsize| #2}
\DeclarePairedDelimiterX\cbraket[2]{(\!(}{)\!)}{#1 \delimsize\| #2}
\DeclarePairedDelimiterX\ketbra[2]{\lvert}{\rvert}{#1 \delimsize\rangle\!\langle #2}
\DeclarePairedDelimiterX\kketbra[2]{|}{|}{#1 \delimsize\rangle\!\rangle\!\langle\!\langle #2}
\DeclarePairedDelimiterX\cketbra[2]{\|}{\|}{#1 \delimsize)\!)\!(\!( #2}
\DeclarePairedDelimiterX\inner[2]{\langle}{\rangle}{#1,#2}

\def\tr{{\rm tr}\,}

\theoremstyle{plain}
\newtheorem{theorem}{Theorem}
\newtheorem{prop}{Proposition}
\newtheorem{lemma}{Lemma}

\theoremstyle{definition}
\newtheorem{definition}{Definition}

\theoremstyle{remark}

\newcommand{\arxiv}[2][]{\ifthenelse{\isempty{#1}}{\href{http://arxiv.org/abs/#2}{{\tt arXiv:\allowbreak{}#2}}} {\href{http://arxiv.org/abs/#2}{{\tt arXiv:\allowbreak{}#2 [#1]}}}}

\newcommand{\booktitle}{\textsl}
\newcommand{\hrefdoi}[2]{\href{https://dx.doi.org/#1}{#2}}

\begin{document}
\title{Discrete Wigner Functions from Informationally Complete Quantum Measurements} \author{John B.\ DeBrota}
\author{Blake
  C.\ Stacey} \affil{\small
  \href{http://www.physics.umb.edu/Research/QBism}{QBism Group},
  Physics Department, University of Massachusetts Boston, \par 100
  Morrissey Boulevard, Boston MA 02125, USA}
\date{\today}

\maketitle

\begin{abstract}
    Wigner functions provide a way to do quantum physics using quasiprobabilities, that is, “probability” distributions that can go negative. Informationally complete POVMs, a much younger subject than phase space formulations of quantum mechanics, are less familiar but provide wholly probabilistic representations of quantum theory. In this paper, we show that the Born Rule links these two classes of structure and discuss the art of interconverting between them. In particular, we demonstrate
    that the operator bases corresponding to minimal discrete Wigner functions (Wigner bases) are orthogonalizations of minimal informationally complete measurements (MICs). By not imposing a particular discrete phase space structure at the outset, we push Wigner functions to their limits in a suitably quantified sense, revealing a new way in which the symmetric informationally complete measurements (SICs) are significant. Finally, we speculate that astute choices of MICs from the orthogonalization preimages of Wigner bases may in general give quantum measurements conceptually underlying the associated quasiprobability representations.
\end{abstract}

\section{Introduction}
In the practical course of doing physics, managing coordinate systems
is an important skill. It is helpful to know how to choose a basis
that makes a problem as simple as possible, and it is beneficial to
understand what can and cannot be eliminated by a clever choice of
reference frame. ``One good coordinate system may be worth more than a
hundred blue-in-the-face arguments,'' a colleague advises
us~\cite{Fuchs:2011}. To that end, this article will explore two
particular classes of normalized operator bases and the relations
between them.

We will work in the quantum theory of finite-dimensional systems
familiar from the study of quantum information and
computation~\cite{Nielsen:2010}. In this theory, each physical system
is associated with a Hilbert space $\mathcal{H}_d \simeq
\mathbb{C}^d$, where the dimension $d$ can be taken as a physical
characteristic of the system. For example, in a quantum computer
containing $N$ \emph{qubits,} the dimension is $d = 2^N$. A
\emph{quantum state,} the means of expressing the preparation of a
system, is an operator on this Hilbert space that is positive
semidefinite and has a trace of unity.  Measurements that one can
perform upon a system are represented mathematically as
\emph{positive-operator-valued measures} (POVMs), which are
resolutions of the identity into positive semidefinite operators:
\begin{equation}
\sum_{i=1}^n E_i = I.
\end{equation}
Each \emph{effect} $E_i$ corresponds to a possible outcome of the
measurement $E=\{E_i\}$, and the probability of that outcome is
calculated by the \emph{Born Rule}:
\begin{equation}
    p(E_i) = \tr (\rho E_i).
\end{equation}
That is, probabilities are Hilbert--Schmidt inner products between the
operators that stand for outcomes and for preparations (or
\emph{priors,} in more probabilistic language). If the set of effects
$\{E_i\}$ spans $\mathcal{L}(\mathcal{H}_d)$, the space of linear operators on $\mathcal{H}_d$, then any operator in this space can be expressed as a list of inner products. Such a
POVM is \emph{informationally complete,} since any quantum-mechanical
calculation about the system may be done in terms of these inner products. In order to
be informationally complete, a POVM must span $\mathcal{L}(\mathcal{H}_d)$, so it
must contain at least $d^2$ effects. An informationally complete POVM
with exactly $d^2$ effects is a \emph{minimal informationally
  complete} measurement, or MIC. 

Because a MIC is informationally complete, it can serve
as a \emph{reference measurement.} This means that any MIC has the
property that if an agent has written a probability distribution over
its $d^2$ possible outcomes, she can then compute the probabilities that
she should assign to the outcomes of any other measurement; in other words, a MIC measurement allows us to think of the Born Rule as furnishing a fully probabilistic representation of quantum theory. As we and collaborators have emphasized in the past, and will revisit below, proofs of the impossibility of probabilistic representations of quantum theory should actually be understood as proofs of the impossibility of probabilistic representations \emph{where the probabilities combine in a particular way inspired by
classicality}\cite{DeBrota:2018}. When using a MIC $E$ as a reference measurement, one must be mindful of its \emph{bias} $\{e_i\}$, composed of the \emph{weights} $e_i:=\tr E_i$ ; if the weights are equal, then $e_i=1/d$ and we call the MIC \emph{unbiased}. The bias of a MIC is an informational characteristic of the measurement --- it is equal to $d$ times the Born Rule probabilities for the MIC measurement given the quantum state of maximal indifference, the \emph{garbage
state}, $\rho=\frac{1}{d}I$. 

Much of the community's interest in MICs has focused upon the
fascinating special case of the \emph{symmetric informationally
  complete} measurements, the SICs~\cite{Zauner:1999, Renes:2004,
  Scott:2010a, Fuchs:2017a}. A SIC is an unbiased MIC where each
effect is proportional to a rank-1 projector --- so, each outcome of
the measurement is specified by a ray in the Hilbert space --- and the
inner products between any two effects are constant. This latter condition is captured by the Gram matrix, that is, the matrix of inner products $[G]_{ij}:=\tr E_iE_j$, of a SIC taking the particularly simple form
\begin{equation}
    [G_{\rm SIC}]_{ij}= \frac{1}{d^2} \frac{d\delta_{ij} + 1}{d+1}.
\end{equation}
SICs have proved in many ways optimal among MICs~\cite{Fuchs:2003,
  Scott:2006, Fuchs:2013, Appleby:2014b, Appleby:2015, DeBrota:2018}. SICs will once again occupy a privileged position from the perspective taken in this paper.

Projective measurements (i.e.,~those corresponding to projections onto the eigenspaces of quantum observables) are not informationally complete; probabilities corresponding to the outcomes of such measurements only provide a partial picture of one's expectations for all possible measurements. It seems the feeling that projective measurements are nevertheless the most conceptually significant variety combined with the development of quasiprobabilistic representations with powerful phase space
analogies produced in some the intuition that probability theory itself was insufficient or at least inconvenient to handle the oddities of quantum mechanics. The quantum foundations, information and computation communities have, accordingly, often looked for deviations from classical behavior in the peculiarities of quasiprobability distributions (e.g.,~the appearance of negativity)~\cite{Ferrie:2011,Veitch:2012,Veitch:2014,Zhu:2016}. Discrete Wigner function approaches in particular have received substantial theoretical and applied interest. Our purpose in this paper is to connect reference measurement based probabilistic representations with a suitably generalized notion of discrete Wigner function
representations. Our hope is that a deeper understanding of this association will be valuable to both probabilistic and quasiprobabilistic points of view and ultimately help us better understand what is possible in a quantum world.

Our plan for this paper is as follows. In
\S\ref{sec:prob}, we will use the MIC concept to introduce
probabilistic and quasiprobabilistic representations of quantum theory and identify the intuition that there should be a natural relation between them. In \S\ref{sec:wf} we consider a generalization of a MIC we call a measure basis and use it to define discrete minimal Wigner bases, setting the scene for a formal discussion of the intuitive relation noted earlier. \S\ref{sec:frame} is a brief interlude providing background and motivation from frame theory; in particular we introduce the frame operator and the
canonical tight frame. With the background in place, in
\S\ref{sec:pwf} we define the principal Wigner basis and derive some of its first properties, such as an induced equivalence class on the set of MICs. In \S\ref{sec:examples} we begin to apply the perspective gained in the previous section to MICs and Wigner bases appearing in the literature. Here, we will prove
Theorem~\ref{distanceToQreps}, which bounds the distance between an
unbiased MIC and an unbiased Wigner
basis. Theorem~\ref{sics-are-best-again} then captures the special
role that SICs play in these
considerations. In the broader picture of our research program, Theorem~\ref{sics-are-best-again} is the key result of the paper, for it demonstrates a new way in which SICs are extremal among MICs. Finally, in \S\ref{sec:discussion} we discuss what has been learned, record a few open questions, and anticipate a few fruitful directions for further research.

\section{Probability and Quasiprobability}
\label{sec:prob}
As we noted above, a MIC may serve as a reference measurement. We can illustrate the meaning of this by comparing and contrasting it
with classical particle mechanics. There, a ``reference measurement''
would just be an experiment that reads off the system's phase-space
coordinates, i.e., the positions and momenta of all the particles
making up the system.  Any other experiment, such as observing the
total kinetic energy, is in principle a coarse-graining of the
information that the reference measurement itself provides.

To develop the analogy, consider the following
scenario~\cite{Appleby:2017c}.  An agent Alice has a physical system
of interest, and she plans to carry out either one of two different,
mutually exclusive laboratory procedures upon it.  In the first
protocol, she will drop the system directly into a measuring apparatus
and thereby obtain an outcome.  In the second protocol, she will
cascade her measurements, sending the system through a reference
measurement and then, in the next stage, feeding it into the device
from the first protocol.  Probability theory in the abstract
\emph{provides no consistency conditions} between Alice's expectations
for these two protocols.  Different circumstances, different
probabilities!  Let $P$ denote her probability assignments for the
consequences of following the two-step procedure and $Q$ those for the
single-step protocol.  Then, writing $\{H_i\}$ for the possible
outcomes of the reference measurement and $\{D_j\}$ for those of the
other,
\begin{equation}
  P(D_j) = \sum_i P(H_i) P(D_j|H_i).
\end{equation}
This much is just logic, or more specifically speaking, a consequence
of Dutch-book coherence~\cite{Fuchs:2013, Diaconis:2017}.  It is known
as the Law of Total Probability (LTP).  \emph{However,} the claim that
\begin{equation}
  Q(D_j) = P(D_j)
\end{equation}
is an assertion of \emph{physics,} above and beyond probabilistic
self-consistency.  It codifies in probabilistic language the idea that
the classical ideal of a reference measurement simply reads off the
system's pre-existing phase-space coordinates, or data equivalent thereto.

In quantum physics, life is very different.  Instead of taking a
weighted average of the $\{P(D_j|H_i)\}$ as in the LTP, Alice instead
uses a mapping
\begin{equation}
  Q(D_j) = \mu\left(\{P(H_i)\}, \{P(D_j|H_i)\}\right),
\end{equation}
where the exact form of the function $\mu$ depends upon her choice of
MIC.  Conveniently, quantum theory is only so nonclassical that $\mu$
is a bilinear form, rather than a more convoluted function.

We can write our equations more compactly by introducing a vector
notation, in which omitted subscripts imply that an entire vector or
matrix is being treated as an entity.  Then the LTP has the expression
\begin{equation}
  P(D) = P(D|H)P(H)\;,
\end{equation}
while the quantum relation, the Born Rule, is\footnote{To obtain \eqref{ltpanalog}, expand an arbitrary state $\rho$ in the basis $\{\rho_i\}$, $\rho_i=H_i/h_i$, of quantum states proportional to the MIC basis:  $\rho=\sum_i\alpha_i\rho_i$. Then, computing the MIC probabilities for $\rho$ with the Born Rule, it follows that $\alpha=\Phi P(H)$. Thus, since  $P(D_j|H_i)=\tr D_j\rho_i$, \eqref{ltpanalog} follows from another application of the Born Rule with the $D$ measurement.} \begin{equation}\label{ltpanalog}
  Q(D) = P(D|H)\Phi P(H)\;, \hbox{ with }
  [\Phi^{-1}]_{ij} := \frac{1}{\tr H_j} \tr H_i H_j\;.
\end{equation}
The Born matrix $\Phi$ depends upon the MIC\footnote{And, in general, a set of post-measurement states; in this case the post-measurement states $\rho_i$ are proportional to the MIC itself and so do not constitute another dependency. We considered the more general update procedure extensively in \cite{DeBrota:2018}.}
, but it is always a column
quasistochastic matrix, meaning its columns sum to one but may contain
negative elements~\cite{DeBrota:2018}. In fact, $\Phi$ \emph{must}
contain negative entries; this follows from basic structural
properties of quantum theory~\cite{Ferrie:2011}.  As a consequence,
$\Phi P(H)$ is a quasiprobability. Considering the operation of $\Phi$
on $P(H)$ as a single term results in an equation algebraically
equivalent to the LTP aside from the appearance of negativity in the
last term. The same thing happens if we regard $\Phi$ as acting to the
left on $P(H|D)$, but now the negativity has been relegated to the
first term. For an unbiased MIC, $\Phi$ is a symmetric matrix and we can do the same thing in yet another way by
``splitting the map down the middle'' and acting with one factor in
either direction:
\begin{equation}\label{downthemiddle}
    Q(D)=\left(P(D|H)\Phi^{1/2}\right)\left(\Phi^{1/2}P(H)\right)\;,
\end{equation}
where the principal square root $\Phi^{1/2}$, which also turns out to be quasistochastic, provides a similar quasiprobabilistic interpretation to $\Phi^{1/2}P(H)$ and the rows of $P(D|H)\Phi^{1/2}$. One might then say that, for a
given MIC, there is a certain ``gauge freedom'' about where the
negativity can occur if we wish to massage \eqref{ltpanalog} into a
form which fits together in exactly the same way as the
LTP~\cite{Stacey:2018b}. As we will see later on, the even-handed gauge choice \eqref{downthemiddle} amounts to a Wigner function representation naturally associated with the reference measurement MIC.

\section{Discrete Minimal Wigner Functions}
\label{sec:wf}
We now begin a more formal discussion of quasiprobability representations
of quantum theory constructed on orthogonal operator
bases~\cite{Zhu:2016a, DeBrota:2017}, which were motivated in the
first place by the demand that we cast the Born Rule structurally analogous to the LTP. 
We first define the following generalization of a MIC:
\begin{definition}
    A \emph{measure basis} is a Hermitian basis $L$ for $\mathcal{L}(\mathcal{H}_d)$ for which 
\begin{equation}
    \sum_iL_i=I \quad {\rm and}\quad \tr L_i\geq 0\;.
\end{equation}

\end{definition}
\noindent The sum condition ensures that its inner products with a quantum state give a quasiprobability in the same way the Born Rule gives probabilities for a POVM. As with a MIC, the bias $\{l_i\}$ of a measure basis consists of the weights $l_i:=\tr L_i$. The condition that the weights are nonnegative permits us to carry over the probabilistic significance of a MIC's bias to a measure basis in general. It will also sometimes be useful to define the diagonal bias matrix
$[A]_{ij}:=l_i\delta_{ij}$ of a measure basis. We extend the Born matrix definition to any measure basis in the natural way: In terms of its Gram matrix and bias matrix, $\Phi=AG^{-1}$. Unless otherwise specified, a measure basis has a bias with weights denoted by the lowercase letter equivalent of the basis elements, e.g.\ the bias of $B$ is $\{b_i\}$. In this language, a MIC is a measure basis consisting of positive semidefinite operators.

A MIC cannot be
orthogonal\cite{DeBrota:2018b}, but a measure basis can --- this is the distinguishing property of the operator bases corresponding to discrete Wigner functions which
presently concerns us. Accordingly, we define
\begin{definition}
    A \emph{discrete minimal Wigner basis} is an orthogonal measure basis.
\end{definition}

\noindent We will generally omit the additional qualifiers ``discrete'' and ``minimal'' since we are only considering finite dimensions and we will not be
discussing quasiprobability representations of quantum theory that use
overcomplete operator bases.\footnote{For a generalization of Wigner functions which does use overcomplete bases, see~\cite{Zurel:2020}.} Up to proportionality, what we call a Wigner basis is the set of ``phase point operators'' for most Wigner function approaches in the literature. We prefer to use phase space agnostic terminology as we have intentionally avoided building any particular conception of a discrete phase space into our generalization. This move will hopefully allow us to
eventually understand precisely when, and to what extent, such conceptions are most useful. For the remainder of the paper, when we speak of a
Wigner function, we mean the quasiprobability distribution one may obtain from a Wigner basis and a density matrix. There is not much to say about Wigner bases at this level of generality, so for now we note two basic properties before proceeding. 

\begin{prop}\label{wfgram}
    The Gram matrix of a Wigner basis is equal to its bias matrix, $[G]_{ij}=f_i\delta_{ij}$. 
\end{prop}
\begin{proof}
    Let $\{F_i\}$ be a Wigner basis. As it is orthogonal, the Gram matrix is diagonal, that is, $\tr
F_iF_j = c_i\delta_{ij}$ for some constants $\{c_i\}$. The dual basis is the 
unique basis $\{\widetilde{F}_i\}$ such that $\tr\widetilde{F}_iF_j =
\delta_{ij}$, so we must have $F_i = c_i\widetilde{F}_i$. Then the sum constraint enforces
$\tr\widetilde{F}_i = 1$ and $c_i=f_i$ follows.
\end{proof}

\begin{definition}
    Let $F$ be a Wigner basis. The \emph{shifted Wigner basis} of $F$ is the set $F^{\rm S}$,
    \begin{equation}
        F^{\rm S}_i:=-F_i+\frac{2f_i}{d}I\;.
    \end{equation}
\end{definition}
\begin{prop}
    The shifted Wigner basis is a Wigner basis with the same bias.
\end{prop}
\begin{proof}
    $\tr F_i^{\rm S}=f_i$ and $\sum_iF_i^{\rm S}=I$ are obvious from inspection. Then
\begin{equation}
  \tr\left(-F_i+\frac{2f_i}{d}I\right)\!\left(-F_j+\frac{2f_j}{d}I\right)
  = \tr F_iF_j
  - \frac{2f_j}{d}\tr F_i
  - \frac{2f_i}{d}\tr F_j
  + \frac{4f_if_j}{d^2}\tr I
  = f_i\delta_{ij}\;
\end{equation}
demonstrates orthogonality.
\end{proof}

\section{The Frame Operator and the Canonical Tight Frame}
\label{sec:frame}

In addition to the Gram matrix, another important operator for classifying sets of vectors is the \emph{frame operator} which appears in the theory of frames. For an introduction to frame theory, we recommend the reference~\cite{Waldron:2018}. In the finite dimensional setting, a frame for a Hilbert space is a spanning set of vectors; a basis is a frame with the minimal number of vectors to be a spanning set. Specializing to $\mathcal{L}(\mathcal{H}_d)$, we define
\begin{definition}
    The \emph{frame operator} of a frame $L$ is the superoperator $S: \mathcal{L}(\mathcal{H}_d)\to\mathcal{L}(\mathcal{H}_d)$ defined by
    \begin{equation}
        X\mapsto \sum_i (\tr XL_i)L_i\;,\qquad \forall X\in \mathcal{L}(\mathcal{H}_d)\;.
    \end{equation}
\end{definition}

\noindent The frame operator is self-adjoint with respect to the Hilbert--Schmidt inner product and has the same nonzero spectrum as the Gram matrix; for a basis, such as a MIC, they are isospectral. By inspection of the definition, one sees that the frame operator allows us to construct and interconvert between the frame and the \emph{dual}, $\widetilde{L}_i := S^{-1}(L_i)$; for a basis, the dual frame is exactly the dual basis. It is not hard to show that the inverse of the frame operator is the frame operator of the dual. Likewise, for bases, the Gram matrix is also invertible and its inverse is given by the Gram matrix of the dual basis
\begin{equation}\label{Ginv}
    [G^{-1}]_{ij}=\tr \widetilde{L}_i\widetilde{L}_j\;.
\end{equation}
We will also make use of the fact that the inner products of a vector with a frame give the expansion coefficients for expressing the vector in the dual frame:
\begin{equation}\label{resid}
    X = \sum_i (\tr X \widetilde{L}_i)L_i=\sum_i (\tr X L_i)\widetilde{L}_i\;,\qquad \forall X\in\mathcal{L}(\mathcal{H}_d)\;.
\end{equation}

When the frame operator for a set of vectors is proportional to the identity, a frame is called \emph{tight}; and when the constant of proportionality is unity, a tight frame is called \emph{normalized}. For all frames there is a naturally associated normalized tight frame known as the \emph{canonical tight frame} which is ``halfway'' to the dual, obtained by applying the inverse square root of the frame operator to each vector. Intuitively, halfway between a basis and its dual is an orthonormal basis, that is, a self-dual basis. Indeed,
for bases, the canonical tight frame construction is an orthogonalization procedure corresponding to a symmetric version of the
Gram--Schmidt algorithm \cite{Waldron:2018}. The canonical tight frame is distinguished by being the \emph{closest} tight frame to the original frame in the sense of minimizing the least squares error (Theorem 3.2 in \cite{Waldron:2018}). 

A MIC cannot be a tight frame\footnote{This is a good place to clear up a confusion that we have encountered
a few times during conferences, when people from different subfields
try to communicate.  A MIC is a basis for the $d^2$-dimensional space
$\mathcal{L}(\mathcal{H}_d)$. It is not overcomplete, but exactly
complete, having just the right number of elements to span the
operator space while keeping itself a linearly-independent set.  A
rank-1 MIC, for which $E_i = e_i \ketbra{\psi_i}{\psi_i}$, is
specified by a set of weights $\{e_i\}$ and a set of vectors
$\{\ket{\psi_i}\}$.  These vectors are $d^2$ in number and live within
$\mathcal{H}_d$, so for \emph{that}
space, they would be overcomplete. The vectors $\{\sqrt{e_i}\ket{\psi_i}\}$ may be considered a frame for the
$d$-dimensional space $\mathcal{H}_d$, and, in fact, this is a normalized tight frame because the frame operator and POVM sum condition coincide: $S=\sum_ie_i\ketbra{\psi_i}{\psi_i}=I$. As an operator basis, however, a MIC is never a tight frame.} --- if its
frame operator were proportional to the identity, its Gram matrix
would have to be as well and a MIC cannot be an orthogonal
basis. For the same reason, an unbiased Wigner basis \emph{is} a tight frame. The canonical tight frame for a MIC $\{E_i\}$ is the orthonormal operator basis
$\{S^{-1/2}(E_i)\}$. This cannot be a measure basis because the sum condition fails for a normalized operator basis, but if the MIC is unbiased, $1/\sqrt{d}$ times the canonical tight frame is an unbiased Wigner basis. In what follows, we will extend this observation to MICs and Wigner bases of arbitrary bias and begin to explore the consequences. 

\section{The Principal Wigner Basis}
\label{sec:pwf}

In this section we define the principal Wigner basis which is the most significant association of a Wigner basis to a given MIC. We will see how the principal Wigner basis induces an equivalence class among measure bases and MICs more specifically. We then begin to study the structure of that class.

An arbitrary Wigner basis is not a tight frame; for bases this is because the tight frame concept effectively demands orthogonality \emph{and} that the vectors all have the same norm. From Proposition \ref{wfgram}, we see that the Hilbert--Schmidt norm of the Wigner basis elements are the square roots of the weights. So for any Wigner basis, if we divide each element by the square root of its weight, the result is a normalized tight frame. We would like to think of Wigner
bases playing the analogous role among measure bases that the canonical tight frame plays among frames. This motivates the following modification of the frame operator:

\begin{definition}\label{rescaledS}
    The \emph{rescaled frame operator} $\mathcal{S}_L$ of a measure basis $L$ is the frame operator of the \emph{rescaled basis} $\left\{\frac{1}{\sqrt{l_i}}L_i\right\}$. For $X\in\mathcal{L}(\mathcal{H}_d)$,
\begin{equation}
    \mathcal{S}_L(X)=\sum_j\left(\tr X\frac{L_j}{\sqrt{l_j}}\right)\frac{L_j}{\sqrt{l_j}}=\sum_j\left(\frac{\tr XL_j}{l_j}\right)L_j\;.
\end{equation}
\end{definition}
\noindent For all measure bases, $\mathcal{S}_L(I)=I$ and, consequently, $\tr \mathcal{S}_L(X)=\tr \mathcal{S}_L(X)I=\tr X\mathcal{S}_L(I)=\tr X$. In accordance with the motivation above, $\mathcal{S}_F$ is the identity superoperator for a Wigner basis $F$, which follows from $F_i=f_i\widetilde{F}_i$. For a MIC $E$, $\mathcal{S}_E$ is a trace-preserving quantum channel we called the entanglement breaking MIC channel (EBMC) in a previous paper; an EBMC for a rank-1 MIC is also the L\"uders MIC
channel for that MIC~\cite{DeBrota:2019}. It also turns out that the rescaled frame operator for a MIC has independently arisen in a different context
pertaining to quantum state tomography, suggesting a deeper significance~\cite{Zhu:2014}. When acting on one of the MIC elements we have
\begin{equation}
    \mathcal{S}_E(E_i)=\sum_j\left(\frac{\tr E_iE_j}{e_j}\right)E_j=\sum_j[\Phi^{-1}]_{ij}E_j\;,
\end{equation}
from which it follows that 
\begin{equation}
    \mathcal{S}^{-1}_E(E_i)=\sum_j[\Phi]_{ij}E_j\;.
\end{equation}
The action of the Born matrix on the reference probability vector is thus equivalent to the inverse rescaled frame operator on the effects. It is easy to see that $\mathcal{S}_E^{-1}(E_i)$ is the dual basis element to the state $\rho_i$, so the expression \eqref{ltpanalog} simply follows from inserting the identity into the Born Rule:
\begin{equation}
    Q(D_j)=\tr D_j\rho=\sum_i(\tr{D_j\rho_i})\left(\tr \mathcal{S}_E^{-1}(E_i)\rho\right)\;.
\end{equation}
This lets us understand the first type of quasiprobability grouping we discussed; $\tr\mathcal{S}_E^{-1}(E_i)\rho$ is the $i$th element of the quasiprobability vector $\Phi P(E)$. We are now in place to see that the even-handed gauge choice discussed above corresponds to a Wigner function.

\begin{definition}\label{pwb}
    The \emph{principal Wigner basis} of a measure basis $L$ is the set $PW(L)=\{F_i\}$,
    \begin{equation}
        F_i:=\mathcal{S}_L^{-1/2}(L_i)\;.
    \end{equation}
\end{definition}
\begin{prop}\label{pwisWF}
    For a measure basis $L$, $PW(L)$ is a Wigner basis with the same bias as $L$.
\end{prop}
\begin{proof}
    As the identity is an eigenvector with eigenvalue $1$ of $\mathcal{S}_L$, this will also be true of $\mathcal{S}^{-1/2}_L$, so
    \begin{equation}
        \sum_iF_i=\sum_i\mathcal{S}_L^{-1/2}(L_i)=\mathcal{S}_L^{-1/2}(I)=I\;.
    \end{equation}
    The bias is preserved because $\mathcal{S}^{-1/2}_L$ is trace-preserving. Orthogonality follows from the self-adjointness of the rescaled frame operator and the relation between a basis and its dual:
    \begin{equation}
        \tr F_iF_j=\tr\!\left(\mathcal{S}_L^{-1/2}(L_i)\mathcal{S}_L^{-1/2}(L_j)\right)=\tr L_i\mathcal{S}_L^{-1}(L_j)=\sqrt{l_il_j}\tr\!\left( \frac{L_i}{\sqrt{l_i}}\mathcal{S}_L^{-1}\!\left(\frac{L_j}{\sqrt{l_j}}\right)\right)=l_i\delta_{ij}\;.
    \end{equation}
\end{proof}
\begin{prop}\label{MICbasis}
    Let $L$ be a measure basis with Gram matrix $G$ and bias matrix $A$. The square root of the Born matrix with all positive eigenvalues,
    \begin{equation}
    \sqrt{\Phi} := A^{1/2}
  \left(A^{1/2}G^{-1}A^{1/2}\right)^{1/2} A^{-1/2}\;,
    \end{equation}
    gives the expansion coefficients of the principal Wigner basis in the measure basis:
    \begin{equation}
        F_i=\sum_j[\sqrt{\Phi}]_{ij}L_j\;.
    \end{equation}
\end{prop}
\begin{proof}
    Let $D$ be the rescaled basis of $L$. Since $F_i=\mathcal{S}_L^{-1/2}(L_i)=\sum_j(\tr \mathcal{S}_L^{-1/2}(L_i)\widetilde{L}_j)L_j$, we must show $[\sqrt{\Phi}]_{ij}=\tr \mathcal{S}_L^{-1/2}(L_i)\widetilde{L}_j$. Note that
    \begin{equation}
        [A^{1/2}G^{-1}A^{1/2}]_{ij}=\tr \widetilde{D}_i\widetilde{D}_j=\tr\mathcal{S}_L^{-1}(D_i)\widetilde{D}_j\implies \left[\left(A^{1/2}G^{-1}A^{1/2}\right)^{1/2}\right]_{ij}=\tr\mathcal{S}_L^{-1/2}(D_i)\widetilde{D}_j\;.
    \end{equation}
    Thus 
    \begin{equation}
        \begin{split}
            [\sqrt{\Phi}]_{ij}&= \left[A^{1/2}
        \left(A^{1/2}G^{-1}A^{1/2}\right)^{1/2} A^{-1/2}\right]_{ij}=\sqrt{\frac{l_i}{l_j}}\left[\left(A^{1/2}G^{-1}A^{1/2}\right)^{1/2}\right]_{ij}\\
            &=\sqrt{\frac{l_i}{l_j}}\tr\mathcal{S}_L^{-1/2}(D_i)\widetilde{D}_j=\tr\mathcal{S}_L^{-1/2}(L_i)\widetilde{L}_j\;,
        \end{split}
    \end{equation}
    where we used the fact that $\widetilde{D}_j=\sqrt{l_j}\widetilde{L}_j$\;.
\end{proof}
Applied to MICs, Proposition \ref{pwisWF} demonstrates that the principal Wigner basis is in fact a Wigner basis and Proposition \ref{MICbasis} connects it to the probabilistic representations discussed in \S\ref{sec:prob}; given any MIC reference measurement, there is an associated Wigner function formed by dividing the effect of the Born matrix equally among the probabilities and quasiprobabilities in \eqref{ltpanalog}.

The principal Wigner basis map satisfies several of the properties one would desire for a ``principal'' definition. Global unitary conjugation acts equally on a MIC and its principal Wigner basis; for any unitary $U$, if $\{F_i\}$ is the principal Wigner basis of the MIC $\{E_i\}$, $\{UF_iU^\dag\}$ is the principal Wigner basis of the MIC $\{UE_iU^\dag\}$. This follows from the invariance of the Gram matrix under a global unitary conjugation of the MIC. As a consequence, if a MIC is \emph{group
covariant}, i.e., if it can be produced by taking the orbit of an initial element under the action of a group, thereby making $d^2$ elements out of one, its principal Wigner basis is group covariant with respect to the same group. The principal Wigner basis also respects tensor products in the following way. Let $D$ and $E$ be MICs, not necessarily for the same dimensional Hilbert space. The elementwise tensor products of their elements, $\{D_i\otimes E_j\}$, forms a MIC $D\otimes E$ for the
product dimension. The principal Wigner basis for this MIC is equal to the tensor product of the principal Wigner bases for the constituent MICs, that is, $PW(D\otimes E)=PW(D)\otimes PW(E)$. This follows from the observation that the tensor product of the Born matrices for the constituent MICs is equal to the Born matrix of the tensor product MIC.

Definition \ref{pwb} furnishes a map from any MIC to a particular Wigner basis. But which MICs are in the preimage of a particular Wigner basis? As one might expect from an orthogonalization
procedure, there are infinitely many MICs which share a principal Wigner basis. We group those that do into an equivalence class. 
\begin{definition}
If $PW(L)=PW(M)$ for two distinct measure bases $L$ and $M$, we say $L$ and $M$ are \textit{Wigner equivalent}, denoted $L\sim_{\rm W}M$.
\end{definition}
\noindent What can we say about the Wigner equivalence classes? The first thing to realize is that the effects of a MIC of a given bias can have different norms depending on the purity of $\rho_i$ while, by contrast, for a Wigner basis, the norm of the elements is fixed by the bias. This inspires us to guess that MICs which are ``in the same direction'' will be Wigner equivalent. The following few results will make this intuition precise.
\begin{prop}
    Given a measure basis $L$ and a real parameter $t\neq 0$, the set $L^t=\{L^t_i\}$,
    \begin{equation}\label{collinearMB}
        L^t_i:=t L_i+(1-t)\frac{l_i}{d}I
    \end{equation}
    is a measure basis with the same bias as $L$. \end{prop}
\begin{proof}
$\sum_iL^{t}_i=I$ and $\tr L^t_i=l_i$ are obvious. What remains is to prove linear independence.    
    Suppose there is a set $\{\beta_i\}$ not all zero for which $\sum_i\beta_iL^t_i=0$. Then
    \begin{equation}
        \sum_i\beta_iL^t_i=t\sum_i\beta_iL_i+\frac{(1-t)}{d}\left(\sum_i\beta_il_i\right)I=0\implies\sum_i\beta_iL_i=\sum_i\left(\frac{t-1}{td}\sum_j\beta_jl_j\right)L_i\;,
    \end{equation}
    where we used the fact that $\sum_iL_i=I$. Because $\{L_i\}$ is a basis, the expansion coefficients are unique, so $\beta_i=\frac{t-1}{td}\sum_j\beta_jl_j$ for all $i$, that is, the $\beta_i$ are constant, say $\beta_i=\beta$. Then $\sum_i\beta_iL^t_i=\beta\sum_iL^t_i=\beta I=0\implies \beta=0$, but this is a contradiction, so $L^t$ is a measure basis.
\end{proof}
\begin{definition}
    Let $L$ be a measure basis. Any measure basis $L^t$, defined by \eqref{collinearMB}, is \emph{collinear} with $L$. If $t>0$, it is \emph{parallel} to $L$ and if $t<0$, it is \emph{antiparallel} to $L$.
\end{definition}
\noindent Now we can prove that, in fact, all parallel measure bases are Wigner equivalent. Let $SPW(L)$ denote the shifted principal Wigner basis of a measure basis $L$.
\begin{theorem}
  For any measure basis $L$,
   \begin{equation}
       PW(L^t)=
    \begin{cases}
        PW(L), & \text{if}\ t>0 \\
        SPW(L), & \text{if}\ t<0
    \end{cases}\;.
  \end{equation}
\end{theorem}
\begin{proof}
 Because the relation between $L$ and $L^t$ is relatively simple, their Born matrices $\Phi=AG^{-1}$ also end up simply related:
    \begin{equation}
        \Phi^t=\frac{1}{t^2}\Phi+\left(1-\frac{1}{t^2}\right)\frac{1}{d}AJ\;.
    \end{equation}
    In particular note that $\Phi^t$ for two $t$ values with the same magnitude but opposite signs are equal. $\Phi$ is column quasistochastic, so the left eigenvector with eigenvalue $1$ is the all 1s vector; it is easy to show that the corresponding right eigenvector is the vector of weights. $\sqrt{\Phi}$ shares these properties and because of them one may deduce that
    \begin{equation}
        \sqrt{\Phi}AJ=AJ\sqrt{\Phi}=AJ\;.
    \end{equation}
    From this, it is easy to verify that
    \begin{equation}
        \sqrt{\Phi^t}=\frac{1}{|t|}\sqrt{\Phi}+\left(1-\frac{1}{|t|}\right)\frac{1}{d}AJ\;,
    \end{equation}
    where $|t|$ is the absolute value of $t$. Now we may use Proposition \ref{MICbasis} to compute the principal Wigner basis for $L^t$:
    \begin{equation}
        \begin{split}
            F_i^t&=\sum_j[\sqrt{\Phi^t}]_{ij}L^t_j=\sum_j\left[\frac{1}{|t|}\sqrt{\Phi}+\left(1-\frac{1}{|t|}\right)\frac{1}{d}AJ\right]_{ij}\left(tL_j+(1-t)\frac{l_j}{d}I\right)\\
            &=\frac{t}{|t|}\sum_j[\sqrt{\Phi}]_{ij}L_j+\frac{1-t}{|t|d}\sum_j[\sqrt{\Phi}]_{ij}l_jI+\frac{t}{d}\left(1-\frac{1}{|t|}\right)\sum_j[AJ]_{ij}L_j+\left(1-\frac{1}{|t|}\right)\frac{1-t}{d^2}\sum_j[AJ]_{ij}l_jI\\
            &=\frac{t}{|t|}F_i+\left[\frac{1-t}{|t|}+t\left(1-\frac{1}{|t|}\right)+(1-t)\left(1-\frac{1}{|t|}\right)\right]\frac{l_i}{d}I=\frac{t}{|t|}F_i+\left(1-\frac{t}{|t|}\right)\frac{l_i}{d}I\;,
        \end{split}
    \end{equation}
from which the claim follows.
\end{proof}

\noindent Since all parallel measure bases are Wigner equivalent, all parallel MICs are as well. Conveniently, we have also learned that the shifted principal Wigner basis is the principal Wigner basis of the antiparallel MICs. As the next lemma shows, there is always an interval of MICs collinear with any measure basis. This interval corresponds to the intersection of collinear measure bases with the cone of positive semidefinite operators.
\begin{lemma}\label{MICshift}
    Let $L$ be a measure basis. Define $\sigma_i:=L_i/l_i$. Let $\lambda^{\max}_i$ and $\lambda^{\min}_i$ be the maximum and minimum eigenvalues of $\sigma_i$. Then  $L^t$ is a MIC iff $t\neq 0$ and
    \begin{equation}\label{micrange}
        \max_i\left\{\frac{1}{1-d\lambda^{\max}_i}\right\}\leq t\leq\min_i\left\{\frac{1}{1-d\lambda^{\min}_i}\right\}\;.
    \end{equation}
\end{lemma}
\begin{proof}
    As $L^t$ is a measure basis, for it to be a MIC the elements must all be positive semidefinite. Thus we need to show $\lambda_{\rm min}(L^t_i)\geq 0$ for all $t$ in the range \eqref{micrange}. Since $\tr L_i=l_i$, the average eigenvalue of $L_i$ is $l_i/d$, and so $\lambda_{\rm min}(L_i)\leq l_i/d$. Suppose $t>0$. Then
    \begin{equation}
        \lambda_j(L_i^t)=t\lambda_j(L_i)+(1-t)\frac{l_i}{d}\;,
    \end{equation}
    so 
    \begin{equation}
        \lambda_{\rm min}(L_i^t)=t\lambda_{\rm min}(L_i)+(1-t)\frac{l_i}{d}\geq 0\iff t\leq \frac{1}{1-d\lambda_i^{\rm min}}\;.
    \end{equation}
    For this to be true for all $i$, $t$ must be less than the minimum value of the right hand side. The lower bound is similarly obtained by assuming $t<0$ and making the appropriate adjustments. 
\end{proof}
Starting with a particular MIC, we were able to see that all those parallel to it were Wigner equivalent. But there are Wigner equivalent MICs which are not collinear; in fact, the following theorem opens the way for the construction of arbitrarily many Wigner equivalent MICs which are not collinear. 
\begin{theorem}
    Let $F$ be a Wigner basis. For any measure basis $L$, $\mathcal{S}_L^{1/2}(F)$ is a measure basis collinear with MICs in the Wigner equivance class of $F$. 
\end{theorem}
\begin{proof}
    Note $\tr \mathcal{S}_L^{1/2}(F_i)=\tr \mathcal{S}_L^{1/2}(F_i)I=\tr F_i\mathcal{S}_L^{1/2}(I)=\tr F_i=f_i$. Now from Definition \ref{rescaledS} and the proportionality of a Wigner basis to its dual it follows that
    \begin{equation}
        \mathcal{S}_{\mathcal{S}^{1/2}_L(F)}(X)=\sum_j\left(\frac{\tr\mathcal{S}^{1/2}_L(F_j)X}{\tr\mathcal{S}^{1/2}_L(F_j)}\right)\mathcal{S}^{1/2}_L(F_j)=\mathcal{S}^{1/2}_L\left(\sum_j\left(\frac{\tr F_j\mathcal{S}^{1/2}_L(X)}{f_j}\right)F_j\right)=\mathcal{S}_L(X)\;.
    \end{equation}
    From this we see that $PW(\mathcal{S}^{1/2}_L(F))=F$. Then from Lemma \ref{MICshift} we may construct MICs collinear with $\mathcal{S}^{1/2}_L(F)$ in the Wigner equivalence class of $F$.
\end{proof}

We have likely only just begun to scratch the surface of this topic, but we have seen enough to begin applying what we've learned. Before doing so in the next section we remark on the significance of our observations. Wigner equivalence is a very broad grouping. Every MIC furnishes a distinct probabilistic representation of the Born Rule, but the probabilistic representations obtained from any two Wigner equivalent MICs are related to the same quasiprobabilistic
representation in the way we have seen. In other words, a Wigner function representation of quantum theory is insensitive to differences among the informationally complete measurements consistent with it. This is at least in part an artifact of our level of treatment. For our purposes here, any MIC sufficed to furnish a probabilistic representation, but this is not to say that any MIC would be practically useful in that capacity. Wigner equivalence alone preserves bias,
but, it seems, little else. A good reference measurement should unmask features of quantum state space that were otherwise hard to detect. It seems likely that the reference measurements standing a chance of being useful in this regard would possess mathematical properties beyond the basic definition of a MIC, perhaps being rank-1 or symmetric with respect to a particular group. Indeed, as the phase point operators of most Wigner function approaches are covariant with respect to the
Weyl--Heisenberg group, furnishing discrete analogs of position and momentum operators, we speculate that Weyl--Heisenberg covariant MICs are a class of reference
measurements which provide something reminiscent of a phase space without sacrificing operational significance. We hope the general association we've identified will allow for a dramatic sharpening of the correspondence between
probabilistic and quasiprobabilistic representations once refinements of this variety and others are adopted. 
\section{Behind Every Great Wigner Basis is a Great MIC}
\label{sec:examples}

The previous section introduced the principal Wigner basis and the induced relation of Wigner equivalence. We can carry this study further by specializing to \emph{unbiased} Wigner bases, the case of most practical and theoretical interest as, to our knowledge, it encompasses all of the Wigner functions derived from operator bases to date in the literature. Pursuing this line of inquiry will lead us to a new way in which SICs are extremal among MICs. We begin by showing that the principal Wigner basis of an unbiased MIC
inherits the closeness property enjoyed by the canonical tight frame. Because of the sum normalization condition, there is also a \emph{farthest} Wigner basis from a given unbiased MIC, namely the shifted principal Wigner basis.

\begin{theorem}\label{distanceToQreps}
    Let $E$ be an unbiased MIC and $F$ be an unbiased
    Wigner basis. Let $\lambda_k$ be the $k$th eigenvalue of the
    MIC's frame operator $S$. Then
    \begin{equation}\label{micqrepbounds}
      \sum_k\left(\sqrt{\lambda_k}-\sqrt{1/d}\right)^2
      \leq\sum_i|\!|E_i-F_i|\!|^2\leq\sum_k\left(\sqrt{\lambda_k}
      +\sqrt{1/d}\right)^2-\frac{4}{d}\;,
    \end{equation}
    where the lower bound is saturated iff $F=PW(E)$ and the upper bound is saturated iff
    $F=SPW(E)$. 
\end{theorem}
\begin{proof}
    The lower bound follows from Theorem 3.2 in \cite{Waldron:2018}. The key step is the demonstration that for any MIC $E$ and Wigner basis $F$,
    \begin{equation}
      \sum_j\tr E_jF_j
      \leq \frac{1}{\sqrt{d}}\sum_k\sqrt{\lambda_k}\;,
    \end{equation}
with equality iff 
\begin{equation}
F_j=\frac{1}{\sqrt{d}}S^{-1/2}(E_j)\;,
\end{equation}
that is, $F=PW(E)$. For the upper bound,
\begin{equation}
        \begin{split}
            \sum_i|\!|E_i-F_i|\!|^2&=\sum_i\biggr|\!\biggr|E_i+F^{\rm S}_i-\frac{2}{d^2}I\biggr|\!\biggr|^2\\
            &=\sum_i\tr\!\!\left(E_i+F^{\rm S}_i-\frac{2}{d^2}I\right)\!\!\left(E_i+F^{\rm S}_i-\frac{2}{d^2}I\right)\\
            &=\sum_i\left(\tr E_i^2+\tr {F_i^{\rm S}}^2+2\tr E_i{F_i^{\rm S}}^2 -\frac{4}{d^3}\right)\\
            &\leq\sum_k \left(\lambda_k+\frac{1}{d}+\frac{2}{\sqrt{d}}\sqrt{\lambda_k}\right)-\frac{4}{d}\\
            &=\sum_k\left(\sqrt{\lambda_k}+\sqrt{1/d}\right)^2-\frac{4}{d}\;,
        \end{split}
    \end{equation}
    with equality iff $F^{\rm S}=PW(E)$, that is, iff $F=SPW(E)$.
   \end{proof}
\noindent Theorem \ref{distanceToQreps} does not seem to exactly generalize to the biased case; perhaps a different condition is more appropriate for biased measure bases. However, this theorem inspires us to ask the reverse question: What is the closest MIC to a given Wigner basis? This is probably quite hard to answer in general, but it turns out we can answer it in an important special case to which we now turn.

Let $E$ be a SIC, that is,
\begin{equation}
  E_i := \frac{1}{d}\Pi_i, \hbox{ with }
  \tr \Pi_i \Pi_j = \frac{d\delta_{ij} + 1}{d+1}\;.
\end{equation}
One may calculate
\begin{equation}\label{SICQreps}
  [\sqrt{\Phi}_{\rm SIC}]_{ij}=\frac{1}{\sqrt{d}}[G_{\rm SIC}^{-1/2}]_{ij}
   = \sqrt{d+1}\delta_{ij} + \frac{1}{d^2}(1-\sqrt{d+1})\;,
\end{equation}
which gives the following principal and shifted principal Wigner basis:
\begin{equation}\label{sicframes}
    F_j=
    \frac{1}{d}\left(\sqrt{d+1}\right)\Pi_j
     + \frac{1}{d^2}\left(1-\sqrt{d+1}\right)
     {I} \quad \text{and } F^\text{S}_j
     = -\frac{1}{d}\left(\sqrt{d+1}\right)\Pi_j
     + \frac{1}{d^2}\left(1+\sqrt{d+1}\right) {I}\;.
\end{equation}
As the next theorem demonstrates, the smallest \textit{and} largest
that the bounds in Theorem \ref{distanceToQreps} can be both occur when the unbiased MIC is a SIC. The
minimal lower bound result complements the prior discovery that SICs
are the closest MICs can come to being orthogonal bases~\cite{Appleby:2014b}. Consequently,
the next theorem supports the intuition that SICs are the natural
analogues of orthogonal operator bases contained within
the cone of positive semidefinite operators.
\begin{theorem}\label{sics-are-best-again}
    Let $\{E_i\}$ be an unbiased MIC and $\{F_j\}$ be an unbiased Wigner basis. Then
    \begin{equation}\label{sicqrepbounds}
        \frac{d-1}{d}\left(d+2-2\sqrt{d+1}\right)\leq\sum_i|\!|E_i-F_i|\!|^2\leq\frac{d-1}{d}\left(d+2+2\sqrt{d+1}\right)\;,
    \end{equation}
    where the lower bound is saturated iff $\{E_i\}$ is a SIC and $\{F_j\}$ is its principal Wigner basis and the upper bound is saturated iff $\{E_i\}$ is a SIC and $\{F_j\}$ is its shifted principal Wigner basis.
\end{theorem}
\begin{proof}
    The matrix $dG$ is doubly stochastic for an unbiased MIC, so the maximal eigenvalue of $S$ is always $1/d$. Furthermore, because the diagonal entries of an unbiased MIC's Gram matrix are bounded above by $1/d^2$, we also know that $\tr S\leq 1$. It is then straightforward to perform a constrained optimization to see that the bounds in \eqref{micqrepbounds} achieve their extreme values when
    \begin{equation}
        \lambda=\left(\frac{1}{d},\frac{1}{d(d+1)},\ldots,\frac{1}{d(d+1)}\right)\;.
    \end{equation}
    Plugging this spectrum in to \eqref{micqrepbounds} gives the upper and lower bounds in \eqref{sicqrepbounds}. Such a spectrum occurs iff the MIC is a SIC, a fact that is easy to derive from Lemma 1 in \cite{DeBrota:2018}.
\end{proof}
The Wigner bases \eqref{sicframes} were identified by Zhu~\cite{Zhu:2016a}
for a different reason.\footnote{Zhu calls unbiased Wigner bases
  ``NQPRs'' and prefers to report the dual basis elements. In his
  notation, $Q_j^{-}=dF_j$ and $Q_j^{+}=dF_j^\text{S}$.} Given a
Wigner basis, the \emph{ceiling negativity} of a quantum state
$\rho$ is the magnitude of the most negative entry in the
quasiprobability vector that represents $\rho$. Maximizing the ceiling
negativity over all quantum states yields the ceiling negativity of
the Wigner basis. Zhu proved that the principal and shifted
principal Wigner bases associated with a SIC provide,
respectively, the \emph{lower and upper bounds} on the ceiling
negativity over all unbiased Wigner bases in dimension $d$. Our 
orthogonalization procedure sets Zhu's result in a broader
conceptual context: Zhu's Wigner bases are the output of applying
to a SIC a procedure that works for any MIC. Our quite general definition of a Wigner basis was partly inspired by Zhu's approach. His relaxation of the requirement of a discrete phase space interpretation for his Wigner bases allowed him to propose quasiprobability representations extremizing the computational resource of negativity beyond what would have been possible within a narrower scope. We have similarly aimed to impose very few constraints at the outset to see
to what extent quantum theory, thus unrestrained, might offer replacements for our presuppositions.

SICs are exceptional among MICs, so finding them in a Wigner equivalence class as the closest member to the Wigner basis prompts us to postulate in general that the closest MIC or MICs in an equivalence class to their principal Wigner basis may be a quantum measurement of particular conceptual similarity to the Wigner basis. As we alluded earlier, finding representationally significant refinements to the set of principal Wigner basis preimages stands a chance of enriching our
understanding of both Wigner function representations and informationally complete measurements. With this program in mind, we present some initial observations about a few Wigner bases and some candidate MICs ``behind'' them.

The discrete Wigner functions most familiar from the literature are those introduced by Wootters~\cite{Wootters:1987}. He constructs Wigner bases for prime dimensions and uses tensor products of these to form a Wigner basis for any composite dimension. As we noted in the previous section, the same can be done with Wigner equivalent MICs in the component dimensions to form Wigner equivalent MICs in any composite dimension. For $d=2$, all unbiased Wigner bases are equivalent up to
an overall unitary transformation and permutation; in particular, we may view any of them as the principal Wigner basis for some qubit SIC. Thus, as Zhu notes, reproducing the Wootters--Wigner basis is a matter of choosing the proper SIC --- that is, picking a regular tetrahedron with the
correct orientation in the Bloch sphere~\cite{Zhu:2016a}. Similarly, Wootters' qutrit Wigner basis is exactly the shifted principal Wigner basis for a special SIC in dimension 3, the Hesse SIC~\cite{Zhu:2016a,Stacey:2016c}. More generally, MICs parallel to the remaining necessary Wootters--Wigner bases, those in odd prime dimensions, were first constructed by Appleby~\cite{Appleby:2007}. The Appleby MIC can be constructed in any odd dimension $d$. The elements
$\{E_{k,l}\}$ of this MIC are labeled by ordered pairs of integers
$k,l \in \{0, \ldots, d-1\}$, and each element has rank
$(d+1)/2$. Together, the elements of the Appleby MIC comprise an orbit
under the action of the Weyl--Heisenberg group. Like a SIC, the Appleby MIC is equiangular. In dimension $3$, the Appleby MIC is apparently the MIC antiparallel to the Hesse SIC with the minimal $t$ value in \eqref{micrange}. Nice properties like equiangularity, relatively low rank elements, and the group covariance suggest these MICs may be the most significant MICs Wigner equivalent with Wootters--Wigner bases. 

Another example may be found in a discrete extension of the Cahill--Glauber formalism~\cite{Ruzzi:2005}. The authors there define an orthogonal operator basis in odd dimensions which takes operators to functions on a discrete phase space $\{\mu,\nu\}$. One can form an unbiased Wigner basis from their notation via $F_{\mu,\nu}=\frac{1}{d}\mathbf{T}^{(0)}(\mu,\nu)$, where the unitary operators $\mathbf{U}$ and $\mathbf{V}$ in their definition are the Weyl--Heisenberg ``shift'' and ``phase''
operators, respectively. The parallel, equiangular MIC with largest $t$ value may again be a good choice of associated reference measurement, but this choice is not as compelling as in the Appleby case because the rank of its
effects is $d-1$. Equiangularity may not be worth such a large rank tradeoff. We speculate that there is a better association among the Wigner equivalent MICs to this family of Wigner bases.

A case of particular interest for quantum computation is $N$-qubit systems. The Wootters--Wigner basis for such a system is the $N$-fold tensor product of the qubit Wootters--Wigner basis. Probably the most significant Wigner equivalent MIC to this Wigner basis would be the $N$-fold tensor product of the appropriately oriented qubit SIC. Let $\{E_i: i = 1, \ldots, 4\}$ be a qubit SIC. Up
to the weighting factor, each effect is a rank-1 projector, and the
set of four such projectors can be portrayed as a regular tetrahedron
inscribed in the Bloch sphere~\cite{Renes:2004}.  A \emph{tensorhedron
  MIC}~\cite{DeBrota:2018b} is a POVM whose elements are tensor
products of operators chosen from the qubit SIC $\{E_i\}$:
\begin{equation}
  E_{i_1,\ldots,i_N} = E_{i_1} \otimes \cdots \otimes E_{i_N}.
\end{equation}
Up to an overall unitary conjugation, every qubit SIC is covariant
under the Pauli group, and so every tensorhedron MIC has an $N$-qubit
Pauli symmetry. Perhaps thinking about tensorhedron measurements in $N$-qubit computation scenarios will inspire insights that didn't come from their quasiprobability counterparts.

The study of MICs may also suggest directions of research on the Wigner function side. One of the mysteries of the SICs is that, in all known cases, the
SICs are group covariant.  The \emph{definition} of a SIC does not
mention group covariance anywhere --- the only symmetry in it is the
equality of the inner products --- and so the fact that the known SICs
are all group covariant might be a subtle consequence we do not yet
understand, or it might be an accident of convenience.  We do know,
thanks to Zhu, that in prime dimensions, if a SIC is group covariant
then it must be covariant under the Weyl--Heisenberg group
specifically~\cite{Zhu:2010,Zhu:2012}.  This leaves open the cases of
dimensions that are higher prime powers or products of distinct
primes. And in dimension $d = 8$, there exists in addition to the
Weyl--Heisenberg SICs the class of \emph{Hoggar-type SICs,} which are
related to the octonions and are covariant under the three-qubit Pauli
group~\cite{Hoggar:1981, Hoggar:1998, Szymusiak:2015, Stacey:2016,
  Stacey:2016b}.  All of these SICs can be converted to unbiased
Wigner bases in the manner described above, and the resulting
Wigner bases will inherit the group-covariance properties of the
original SICs. Therefore, for a three-qubit system, the construction of principal Wigner
bases from unbiased MICs furnishes three inequivalent
Wigner bases of interest: Wootters, Weyl--Heisenberg, and Hoggar. The Wootters version is distinguished by particularly nice permutation symmetry properties~\cite{Zhu:2016}.

We conclude this section by noting an example of a measure basis property which may be studied for both MICs and Wigner bases in light of the principal Wigner basis concept. In order to study time evolution, Wootters~\cite{Wootters:1987} explores the \emph{triple products} of his Wigner basis, which in
our notation are
\begin{equation}
\Gamma_{jkl} = d^2 \tr (F_j F_k F_l).
\end{equation}
These can of course be defined for any Wigner basis. Of particular
note is the case where the Wigner basis is the principal Wigner
basis of a SIC, because the SIC triple products $\tr \Pi_j \Pi_k \Pi_l$
are remarkable numbers~\cite{Appleby:2015,Appleby:2017b, Appleby:2017d, Kopp:2018,
  Andersson:2019}.  We have that
\begin{equation}
\begin{split}
d^3 \tr(F_j F_k F_l)
 &= \pm(d+1)^{3/2}\, \tr(\Pi_j \Pi_k \Pi_l) \\
 &\ + (1 - \sqrt{d+1}) 
      (\delta_{jk} + \delta_{kl} + \delta_{jl}) \\
 &\ - \frac{1}{d^2\sqrt{d+1}}
    \left(2\sqrt{d+1} + (d+1)(d-2)\right).
\end{split}
\end{equation}
For Wootters' definition of the discrete Wigner basis, the triple
products can be found using the geometry of the finite affine plane
on~$d^2$ points.  Essentially, one takes the triangle formed by three
points in that phase space, and the triple product depends upon the
``area'' of it.  Specifically, the triple product (for the case of odd
prime dimension) is given by the complex exponential
\begin{equation}
  \Gamma_{jkl} = \frac{1}{d} \exp\left(\frac{4\pi i}{d}
  A_{jkl}\right)\;.
\end{equation}
This leads naturally to an interpretation of the triple products in
terms of \emph{geometric phases.}  The larger the enclosed area, the
greater the geometric phase.  The SIC triple products, and thus by
extension those of their associated Wigner bases, have rich
number- and group-theoretic properties~\cite{Stacey:2016,
  Stacey:2016b, Appleby:2017b, Appleby:2017d, Kopp:2018,
  Andersson:2019}.  What these properties imply for the Wigner
bases derived from SICs is largely an open question. For early
results in this vein, see Theorems 6 and 13 of~\cite{Appleby:2015} and also~\cite{Appleby:2011}.

\section{Discussion}
\label{sec:discussion}

A MIC is a basis for the operator space $\mathcal{L}(\mathcal{H}_d)$,
or in other words, a coordinate system for doing quantum mechanics.
Because MIC elements are required to be positive semidefinite, no MIC
can ever be an \emph{orthogonal} basis; the closest that a MIC can
come to orthogonality is by being a SIC~\cite{Appleby:2014b}.  Prior work has shown that this expresses how
much of the oddity of quantum theory is an artifact of coordinates,
versus what is the unavoidable residuum of nonclassicality~\cite{DeBrota:2018}. Other reference measurements may be optimally suited for other purposes, say, potentially, for algorithm design or to account for experimental specifics. If one
abandons direct operational meaning in terms of probabilities,
one can push basis elements outside of the positive semidefinite cone
and achieve orthogonality.  In this paper, we have shown well-defined
procedures for doing so, and we have quantified how far an
orthogonalized basis --- a Wigner basis --- can deviate from the
original MIC in the unbiased case.

In exploring the consequences of the principal Wigner basis definition, we have found ourselves with a number of thus far unresolved questions. Several pertain to what we believe will be a fruitful direction for further research, namely the disambiguation of Wigner equivalent MICs: When is there a rank-1 MIC in an equivalence class? How is the rank of the MIC related to its distance to the principal Wigner basis? Does the principal Wigner basis suggest anything about the operational
significance of a Weyl--Heisenberg covariant reference measurement? While negative quasiprobabilities do not have direct operational
meaning as probabilities do, they can be made meaningful in
combination with additional data.  Of particular relevance is the
discovery that negativity can be a resource for quantum
computation~\cite{Veitch:2012, Veitch:2014}. With a suitable Wigner equivalent MIC, the analog of negativity may be studied in reference probabilities. In the other direction, perhaps one could explore how useful statistical properties which are easily displayed by probabilistic representations are reflected in the principal Wigner function. In pursuing this inverse problem, one potential place to turn is to resource theory, especially in light of a majorization lemma concerning Born matrices we proved in a
previous paper~\cite{DeBrota:2018}. Grasping the variety of Wigner functions, and how they
relate to the most economical of probabilistic representations of
quantum theory, may prove helpful in advancing our understanding of
this intriguing subject.

\bigskip

We thank Marcus Appleby, Christopher Fuchs, and Huangjun Zhu for discussions.  This
research was supported in part by the John Templeton Foundation. The
opinions expressed in this publication are those of the authors and do
not necessarily reflect the views of the John Templeton
Foundation. JBD was supported by grant FQXi-RFP-1811B of the
Foundational Questions Institute.

\end{document}